\documentclass[preprint,11pt]{elsarticle}
\usepackage{amssymb}
\usepackage{natbib}
\usepackage[french,english]{babel}
\usepackage [vscale=0.76,includehead]{geometry} 
\usepackage{relsize}
\usepackage{graphicx,amsmath,fullpage,mathptmx,helvet}
\usepackage[utf8]{inputenc}
\usepackage[T1]{fontenc}
\usepackage{amsthm}
\usepackage{amsfonts} 
\usepackage{amsmath}
\usepackage{stmaryrd}
\usepackage{tikz}
\usepackage{tikz-qtree}
\usepackage{listings}
\usepackage{algorithm} 
\usepackage{algpseudocode} 
\usepackage{caption}
\usepackage{subcaption}
\usepackage{chngcntr}
\usepackage{url}
\usepackage{textcomp}
\usepackage{pgfplots}

\pgfplotsset{width=10cm,compat=1.9}

\theoremstyle{definition}
\newtheorem{definition}{Definition}[section]
\theoremstyle{plain}
\newtheorem{theorem}{Theorem}[section]
\newtheorem{lemma}[theorem]{Lemma}
\newtheorem{notation}[theorem]{Notation}
\newtheorem{property}[theorem]{Property}
\DeclareMathOperator*{\argmax}{arg\,max}

\begin{document}

\begin{frontmatter}

\title{Online bin stretching lower bounds: Improved search of computational proofs}

\author[inst1]{Antoine Lhomme}
\affiliation[inst1]{organization={Univ. Grenoble Alpes, G-SCOP, CNRS},
            city={Grenoble},
            postcode={38000}, 
            country={France}}
\author[inst1]{Olivier Romane}
\author[inst1]{Nicolas Catusse}
\author[inst1]{Nadia Brauner}

\begin{abstract}
Computing lower and upper bounds on the competitive ratio of online algorithms is a challenging question: For a minimization combinatorial problem, proving a competitive ratio for a given algorithm leads to an upper bound. However computing lower bounds requires a proof on all algorithms. This can be modeled as a 2-player game where a strategy for one of the players is a proof for the lower bound. The tree representing the proof can can be found computationally. This method has been used with success on the online bin stretching problem where a set of items must be packed online in $m$ bins. The items are guaranteed to fit into the $m$ bins. However, the online procedure might require to stretch the bins to a larger capacity in order to be able to pack all the items. This stretching factor is the objective to be minimized. 

We propose original ideas to strongly improve the speed of computer searches for lower bound: propagate the game states that can be pruned from the search and improve the speed and memory usage in the dynamic program which is used in the search. These improvements allowed to increase significantly the speed of the search and hence to prove new lower bounds for the bin stretching problem for 6, 7 and 8 bins. 
\end{abstract}

\end{frontmatter}


\section*{Introduction}
Online optimization represent a class of problems where all of the information is not known before constructing a solution; instead, the information is given gradually after some decisions have been made. This may force to make choices that turn out to be sub-optimal later on (see \citet{Borodin} and \citet{FiatAmos1998OATS} for more information on online algorithms). This paper tackles a specific online problem, \textit{online bin stretching}, which may also be seen as a scheduling problem - $Pm|online, known-OPT|C_{max}$ (see \citet{online_scheduling} for more information on online scheduling, and see \citet{OnlineSchedulingSurvey} for a recent survey on semi-online scheduling).

In online bin stretching, introduced by \citet{AZAR200117}, a set of items must be packed into $m$ bins. Items are revealed sequentially: an item must be packed before the next one is revealed. The bins do not have a capacity constraint, i.e. any item can always be packed in any bin regardless of the items already in the bin. The items sent must also satisfy the following constraint: they must fit into $m$ bins of unit size. The objective is to pack the items into the bins in order to minimize the load of the largest bin.

\citet{AZAR200117} proposed the following application of \textit{online bin stretching}: some files are sent one by one in some order to $m$ servers. The only information known is that the files were initially stored in $m$ other servers of unit capacity. How large should the new servers be in order to store all the files? How can we design an algorithm to place the files in the new servers?

Figure~\ref{Intro_example} gives a simple example with 2 bins ($m=2$). First, an item of size $1/2$ is given and placed in the first bin. Then, an item of size $1/2$ is given and placed in the second bin. But now, an item of size $1$ is given. The first bin needs to be \textit{stretched} to a capacity of $3/2$ in order to receive this new item - hence the name bin \textit{stretching}. This set of items is allowed since it could fit into 2 bins of size 1, if the two items of size $1/2$ are packed together.

\begin{figure}\unitlength = .65cm
     \centering
     \begin{subfigure}[b]{0.3\textwidth}

         \begin{picture}(2,6.5)
		\put(-.01,0){\line(0,1){5}}
		\put(3.01,0){\line(0,1){5}}
		\put(6.03,0){\line(0,1){5}}
		\put(-.01,0){\line(1,0){6.04}}
		\multiput(.05,3)(.2,0){30}{\line(1,0){.1}}
		\put(6.3,2.5){\makebox(1,1){$1$}}

		\put(-.08,0){\colorbox{yellow!20}{\framebox(3,1.5){$1/2$}}}
	    \end{picture}
         \caption{First item: $1/2$}
         \label{firstitem}
     \end{subfigure}
     \hfill
     \begin{subfigure}[b]{0.3\textwidth}
         \begin{picture}(2,6.5)
		\put(-.01,0){\line(0,1){5}}
		\put(3.01,0){\line(0,1){5}}
		\put(6.03,0){\line(0,1){5}}
		\put(-.01,0){\line(1,0){6.04}}
		\multiput(.05,3)(.2,0){30}{\line(1,0){.1}}
		\put(6.3,2.5){\makebox(1,1){$1$}}

		\put(-.08,0){\colorbox{yellow!20}{\framebox(3,1.5){$1/2$}}}
		\put(2.95,0){\colorbox{green!20}{\framebox(3,1.5){$1/2$}}}
	    \end{picture}
         \caption{Second item: $1/2$}
         \label{seconditem}
     \end{subfigure}
     \hfill
     \begin{subfigure}[b]{0.3\textwidth}
         \begin{picture}(2,6.5)
		\put(-.01,0){\line(0,1){5}}
		\put(3.01,0){\line(0,1){5}}
		\put(6.03,0){\line(0,1){5}}
		\put(-.01,0){\line(1,0){6.04}}
		\multiput(.05,3)(.2,0){30}{\line(1,0){.1}}
		\multiput(.05,4.5)(.2,0){30}{\line(1,0){.1}}
		\put(6.3,2.5){\makebox(1,1){$1$}}
		\put(6.2,4){\makebox(1,1){$3/2$}}

		\put(-.08,0){\colorbox{yellow!20}{\framebox(3,1.5){$1/2$}}}
		\put(2.95,0){\colorbox{green!20}{\framebox(3,1.5){$1/2$}}}
		\put(-.08,1.5){\colorbox{red!20}{\framebox(3,3){$1$}}}
	    \end{picture}
         \caption{Third item: $1$}
         \label{thirditem}
     \end{subfigure}
        \caption{Online bin stretching example}
        \label{Intro_example}
\end{figure}
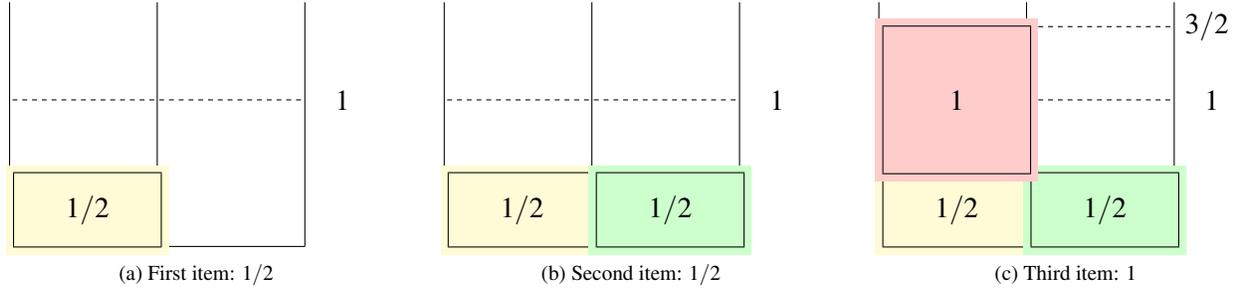

This example shows that it may be necessary to stretch the bins in order to fit all the items. So the question is: what is the smallest size the bins must be stretched to in order to pack any set of items? And how to pack the items into the bins?

\citet{Gabay2017} model this problem as a 2-player game. One player, the adversary, gives items which must be packed into bins by the other player, the algorithm. The algorithm aims to stretch the bins as little as possible, while the adversary tries to give items to force the algorithm to stretch the bins as much as possible. Then, a min-max algorithm can be used to construct a proof of a lower bound on the worst case performance of any online bin stretching algorithm. Using this method, several lower bounds have been proved. For instance, the lower bound 4/3 for any number of bins has been established in 1998 by \citet{AZAR200117} with a "by hand" case study proof; but until 2017, no other lower bounds have been found. \citet{Gabay2017} propose a computational method to find new proofs.  Using this new approach, better lower bounds have been proved (such as 19/14 for $3\leq m \leq 8$). The advantage of the computational approach is to enumerate a large number of different cases. Obtaining even better proofs requires exploring more cases, hence the efficiency of the method is critical in order to improve lower bounds even further.

We restate the idea from \citet{Gabay2017} in Section~\ref{section_LB}; then we propose new ideas to strongly improve the speed of computer searches for lower bounds: Section~\ref{section_propag_alg} presents an idea to propagate information on the game states in order to prune the search; and Section~\ref{next_item} proposes an improvement of the dynamic program of proposed in \citet{Bohm2017}. We also give counting arguments to have an accurate usage of memory during the search in Section~\ref{next_item} and Section~\ref{propag_example}.  With the new ideas proposed, three new lower bounds were found, and the search for lower bounds is faster. Numerical results are detailed in Section~\ref{section-results}.

\section{Online bin stretching}
In the remainder, the number of bins $m$ is fixed.
To give a formalization of the online bin stretching problem, let us first define possible sets of items that can be sent as \textit{instances}. Since the items that are sent can fit into $m$ bins of unit size, we give the following definitions:
\begin{definition}[Bin stretching instance]

A \textit{bin-stretching instance} is a finite sequence of strictly positive numbers $(x_1, \dots, x_k)$ such that there exists a partition $P_1, \dots, P_m$ of the integers $\{1, \dots, k\}$ verifying $$\forall i \in \{1, \dots, m\}\;\;\;\;\sum_{j \in P_i} x_j \leq 1$$
We write $\mathcal{I}_{BS}$ the set of bin-stretching instances, and the empty instance is noted $\emptyset$.
\end{definition}

\begin{definition}[Extension of bin-stretching instances]
A real positive number $y$ is called an \textit{extension} of a bin-stretching instance $I = (x_1, \dots, x_k)$ if the sequence $(x_1, \dots, x_k, y)$ is a bin-stretching instance. 
We note the instance $(x_1, \dots, x_k, y)$ as $I\oplus y$. 
Similarly, a finite sequence of real positive numbers $(y_1, \dots, y_j)$ is an extension of $I$ if the sequence $(x_1, \dots, x_k, y_1, \dots, y_j)$ is a bin-stretching instance. 
\end{definition}
After sending the items of an instance (which can be empty), only items that are extensions of the instance can be sent next. Here is a simple yet useful property on extensions:
\begin{property}\label{smaller_extensions}
Let $I\in\mathcal{I}_{BS}$ be an instance, and $y$ be an extension of $I$. Then any $y'<y$ is also an extension of $I$.
\end{property}
\begin{proof}
Since the items of the instance and the item of size $y$ can fit into $m$ bins of size 1, then the items of the instance and any item of size $y'<y$ can also fit into $m$ bins of size 1.
\end{proof}
An \textit{online bin stretching algorithm} decides in which bin to place a new item. More formally:
\begin{definition}[Online bin stretching algorithm]
An \textit{online bin stretching algorithm} associates an integer between $1$ and $m$ (included) to a triplet $(I, P, y)$, with $I = (x_1, \dots, x_k)$ being a bin-stretching instance, $P = (P_1, \dots, P_m)$ a partition of the integers $\{1,\dots,k\}$, and $y$ being an extension of $I$.
\end{definition}
We write $\mathcal{A}_{BS}$ the set of all bin-stretching algorithms.

\begin{definition}[Result of an algorithm on an instance]
Let $A$ be a bin-stretching algorithm and $I = (x_1, \dots, x_k)$ be a bin stretching instance. We iteratively construct the partitions $P^0, \dots, P^k$ as follows:
\begin{itemize}
    \item $\forall i \in \{1, \dots, m\}\;\;\;\; P_i^0 = \emptyset$
    \item $\forall i \in \{1, \dots, m\}, \forall j \in \{1, \dots, k\}\;\;\;\; P_i^j = \begin{cases}
     P_{i}^{j-1} \cup \{j\} & \text{if } A((x_1, \dots, x_{j-1}), P^{j-1}, x_j) = i\\
     P_{i}^{j-1}  & \text{otherwise}
		  \end{cases}$
\end{itemize}
\end{definition}
The partition $P^j$ corresponds to the arrangement of the first $j$ objects of the instance $I$ according to the algorithm $A$. We can then define the \textit{result} of the algorithm $A$ on the instance $I$, noted $R(I, A)$:
$$R(I,A) = \max_{i\in\{1, \dots, m\}}\sum_{j \in P^k_i}x_j$$
The result of an instance on an algorithm corresponds to the load of the fullest bin after the algorithm places all the objects of the instance. To analyze the performance of algorithms, we look at their worst-case:
\begin{definition}[Stretching-factor of a bin-stretching algorithm]
The \textit{stretching-factor} $s_{BS}(A)$ of a bin-stretching algorithm $A$ is defined by:
$$s_{BS}(A) = \max_{I\in \mathcal{I}_{BS}} R(I, A)$$
\end{definition}

We want to find the best stretching-factor that is achievable by an algorithm, defined as follows:

\begin{definition}[Optimal stretching-factor]
$$s^* = \min_{A\in\mathcal{A}_{BS}} s_{BS}(A) = \min_{A\in\mathcal{A}_{BS}}\;\; \max_{I\in\mathcal{I}_{BS}} R(I, A)$$ 
This lower bound tells us that any online bin stretching algorithm needs, in the worst case, to stretch the bins to a factor of at least $s^*$ to place all the items.
\end{definition}

\section{Finding lower bounds}\label{section_LB}
Computing the exact value of $s^*$ is a difficult task. Therefore, our objective is to find bounds on $s^*$. This paper focus on lower bounds. In order to find such lower bounds, we use similar methods as first introduced by \citet{Gabay2017} and improved by \citet{BOHM20221}. We will first restate the method and then propose improvements. The idea presented in those papers is to construct through computer searches a \textit{tree proof} of a lower bound. 

A \textit{tree proof} describes all possible algorithms on a few well chosen instances:
\begin{definition}[Tree proof]

\begin{enumerate}
    \item A tree $\mathcal{T}$ has a finite number of nodes.
    \item Each node $\mathcal{N} = (I, P, y)$ that is not a leaf corresponds to a bin-stretching instance $I$, a partition $P$ of the objects in the bins and an extension $y$ of $I$.
    \item The root is of the form $\mathcal{R} = (\emptyset, (\emptyset, \dots, \emptyset), y)$.
    \item Each leaf $\mathcal{L} = (I, P)$ corresponds to a bin-stretching instance $I$ and a packing $P$ of the objects in the bins.
    \item The children of a node $\mathcal N = (I, P, y)$, with $|I| = k$, are defined as all the possibilities to place the new item $y$ in a bin, so a node has at most $m$ children. In a more formal manner, for each bin $i$, there exists a child of $\mathcal T$ which is of the form $\mathcal T_i = (I\oplus y, P', y')$ or $\mathcal L_i = (I\oplus y, P')$, with $P_i' = P_i \cup \{k+1\}$ and for any $j\neq i$, $P_j' = P_j$.
    \item We define the \textit{value of a leaf} as the value of the fullest bin, and we define the \textit{value of the tree} $v(\mathcal{T})$ as the minimum over the values of the leaves.
\end{enumerate}
\end{definition}

Finding a tree of a certain value is a lower bound of $s^*$ (see proof below). Figure~\ref{tree_proof_4/3} is an example of such a tree, proving that for $m=2$ bins, $s^*\geq \frac{4}{3}$. For readability, we only write on each node the loads of the bins and the next item that is sent.
\begin{figure}[htbp]
    \centering
    \resizebox{\linewidth}{!}{
	\begin{tikzpicture}
	    \tikzset{every tree node/.style={align=center,anchor=north}}
	    \tikzset{level distance=45pt}
	    \Tree
	    [.{$(0,0)$\\{\color{red} Next: $1/3$}}
	    [.{$(1/3, 0)$\\{\color{red} Next: $1/3$}}
	    [.{$(1/3, 1/3)$\\{\color{red} Next: $1$}}
	    [.{$({\color{blue} 4/3}, 1/3)$} ]
	    [.{$(1/3, {\color{blue} 4/3})$} ]
	    ]
	    [.{$(2/3, 0)$\\{\color{red} Next: $2/3$}}
	    [.{$({\color{blue} 4/3}, 0)$} ]
	    [.{$(2/3, 2/3)$\\{\color{red} Next: $2/3$}}
	    [.{$({\color{blue} 4/3}, 2/3)$} ]
	    [.{$(2/3, {\color{blue} 4/3})$} ]
	    ]
	    ]
	    ]
	    [.{$(0, 1/3)$\\{\color{red} Next: $1/3$}}
	    [.{$(1/3, 1/3)$\\{\color{red} Next: $1$}}
	    [.{$({\color{blue} 4/3}, 1/3)$} ]
	    [.{$(1/3, {\color{blue} 4/3})$} ]
	    ]
	    [.{$(0, 2/3)$\\{\color{red} Next: $2/3$}}
	    [.{$(0, {\color{blue} 4/3})$} ]
	    [.{$(2/3, 2/3)$\\{\color{red} Next: $2/3$}}
	    [.{$({\color{blue} 4/3}, 2/3)$} ]
	    [.{$(2/3, {\color{blue} 4/3})$} ]
	    ]
	    ]
	    ]
	    ]
	\end{tikzpicture}
	}
    \caption{Tree proof of the lower bound 4/3 for 2 bins}
    \label{tree_proof_4/3}
\end{figure}
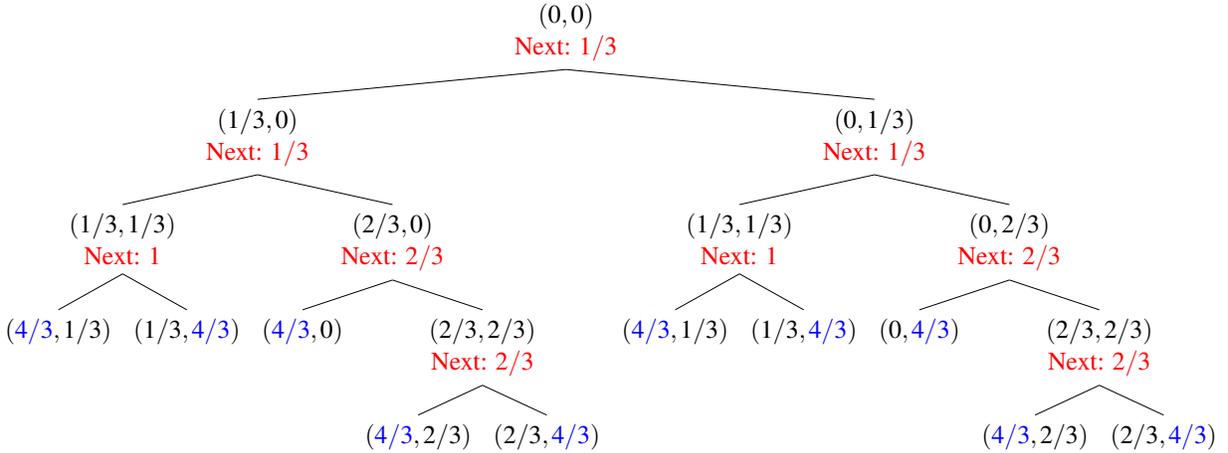

\begin{property}
The value of any tree proof is a lower bound of $s^*$.
\end{property}
\begin{proof}
Let $\mathcal{T}$ be a tree proof of value $v$ and of root $\mathcal{R}$. We will show that $s^*\geq v$.\\
Let $A\in\mathcal{A}_{BS}$ be an algorithm. We construct an instance $I$ such that the result of the algorithm on the instance $R(I, A)$ is at least $v$. To do so, we follow nodes $\mathcal{N}_i = (I_i, P_i, y_i)$ in the tree as follows:
\begin{itemize}
    \item $\mathcal{N}_0 = \mathcal{R}$
    \item If $\mathcal{T}_i = (I_i, P_i, y_i)$ is not a leaf, the algorithm decides to place the item $y_i$ in the bin $A(I_i, y, P_i)$. We then define $\mathcal{T}_{i+1}$ to be the child node of $\mathcal{T}_i$ corresponding to placing the item $y_i$ in the bin $A(I_i, y_i, P_i)$, which always exists by definition of a tree proof.
    \item If $\mathcal{T}_i$ is a leaf, then by construction of the tree, the value of the leaf corresponds exactly to the result of the instance $I_i$ on the algorithm $A$.
\end{itemize}
Since the tree is finite, a leaf $\mathcal{L} = (I, P)$ is ultimately reached, and its value corresponds exactly to the result of the instance $I$ on the algorithm $A$. Since the value of the tree $v(\mathcal{T})$ is defined as the smallest value of a leaf, we have:
$$v(\mathcal{T}) \leq R(I, A) \leq s_{BS}(A)$$
So,
$$\forall A\in\mathcal{A}_{BS}, v(\mathcal{T}) \leq s_{BS}(A)$$
Hence $v(\mathcal{T}) \leq s^*$.
\end{proof}

The stretching factor of any algorithm is an upper-bound of $s^*$. To obtain a lower bound, we can model the problem as a 2-player game, one being the algorithm trying to fit items in the bins, and one being the adversary proposing items. Then we can explore computationally a game tree corresponding to this problem.

In order to deal with integer values and to have a finite game, we scale all the previous definitions by a certain factor $g$, such that all the bins have capacity $g\in \mathbb{N}^*$, and we restrict the item weights to $\{1, \dots, g\}$. Eventually, we increase the value of $g$ so that more items can be proposed, and better bounds can be found.

The goal of the player "adversary" is to propose items such that no matter what the algorithm does, one bin will reach at least a certain value $t$, with $t>g$. If we succeed in finding such a strategy for the player "adversary", then a tree proof of value $t/g$ has been constructed, and we have obtained a lower bound of $s^*$: 
$$s^*\geq\frac{t}{g}$$
This means that any online bin-stretching algorithm has a stretching-factor of at least $\frac{t}{g}$ (as in Figure~\ref{tree_proof_4/3}).

The adversary can only send items such that the items fit into $m$ bins of size $g$. If the adversary cannot send any more items, and no bin is filled with items of total size greater or equal to $t$, then the algorithm wins. If a bin is filled with items of total size greater or equal to $t$ then the adversary wins.

We can then prove a lower bound of $t/g$ on $s^*$ by computing a tree proof of value $t/g$.

So far, the following lower bounds have been found:
\begin{itemize}
    \item \citet{AZAR200117} have established the general lower bound of 4/3 for any number of bins.
    \item \citet{Gabay2017} have found the lower bounds 19/14 for 3 and 4 bins through computer searches of tree proofs.
    \item \citet{BOHM20221} proved the lower bound of 19/14 for 5, 6, 7 and 8 bins and improved the lower bound of 19/14 for 3 bins to 112/82 with an improvement of the previous method.
\end{itemize}

Upper bounds seem more tricky to find, as one needs to construct an algorithm and then analyze the algorithm worst case. So far, the following upper bounds have been found:
\begin{itemize}
    \item \citet{Kellerer1997} found an upper bound of 4/3 for 2 bins - hence the algorithm is optimal.
    \item \citet{AZAR200117} gave a general upper bound of 13/8 = 1.625 for any number of bins.
    \item \citet{AZAR200117} showed that, if the number of bins is between $3$ and $21$, then there exists an algorithm with stretching factor $\frac{5m-1}{3m+1}$.
    \item \citet{Bohm2017} improved the upper bound for 3 bins to 11/8 = 1.375, and improved the upper bound for any number of bins to $3/2$.
    \item \citet{Liesk} improved several upper bounds through computer searches: $31/22$ for 4 bins, $23/16$ for 5 bins, $19/13$ for 6 bins.
\end{itemize}

Table~\ref{bounds_table} and Figure~\ref{graph_bounds_1} give a summary of the known lower bounds and upper bounds on $s^*$. 
\\ \\
\begin{table}
\begin{center}
\begin{tabular}{ |c|c|c|c|c|c| } 
 \hline
 Bins & 2 & 3 & 4 & 5\\ 
 \hline \hline
 Lower bound & $4/3\textbf{ \cite{AZAR200117}}$ & $112/82\textbf{ \cite{BOHM20221}}$ & $19/14\textbf{  \cite{Gabay2017}}$ & $19/14\textbf{  \cite{BOHM20221}}$ \\ 
 \hline
  Decimal & 1.3334 & 1.3659 & 1.3571 & 1.3571 \\
   \hline \hline
   Upper bound & $4/3\textbf{  \cite{Kellerer1997}}$ & $11/8\textbf{  \cite{Bohm2017}}$ & $31/22\textbf{  \cite{Liesk}}$ & $23/16\textbf{  \cite{Liesk}}$ \\ 
 \hline
  Decimal & 1.3334 & 1.375 & 1.409 & 1.438 \\
   \hline
\end{tabular}
\end{center}
\begin{center}
\begin{tabular}{ |c|c|c|c|c| } 
 \hline
 Bins & 6 & 7 & 8 & 9 $\dots \infty$ \\ 
 \hline \hline
 Lower bound & $19/14\textbf{ \cite{BOHM20221}}$ & $19/14\textbf{ \cite{BOHM20221}}$ & $19/14\textbf{ \cite{BOHM20221}}$ & $4/3\textbf{ \cite{AZAR200117}}$ \\ 
 \hline
 Decimal & 1.3571 & 1.3571 & 1.3571 & 1.3334 \\
   \hline \hline
    Upper bound & $19/13\textbf{  \cite{Liesk}}$ & $3/2\textbf{ \cite{Bohm2017}}$ & $3/2\textbf{ \cite{Bohm2017}}$ & $3/2\textbf{ \cite{Bohm2017}}$ \\ 
 \hline
  Decimal & 1.462 & 1.5 & 1.5 & 1.5 \\
   \hline
\end{tabular}
\end{center}
    \caption{Known lower and upper bounds on the bin stretching factor}
    \label{bounds_table}
\end{table}

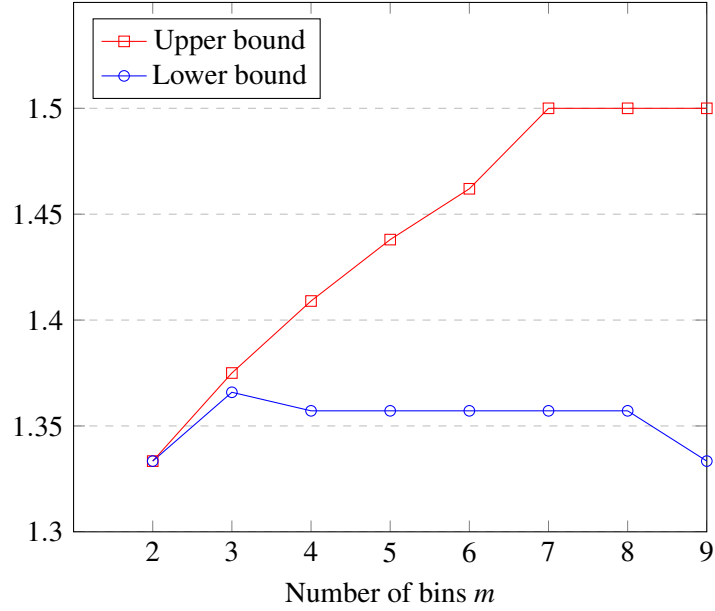
\begin{figure}[htbp]
\centering
\begin{tikzpicture}
\begin{axis}[
    xlabel={Number of bins $m$},
    xmin=1, xmax=9,
    ymin=1.3, ymax=1.55,
    xtick={2, 3, 4, 5, 6, 7, 8, 9},
    ytick={1.3, 1.35, 1.4, 1.45, 1.5},
    legend pos=north west,
    ymajorgrids=true,
    grid style=dashed,
]

\addplot[
    color=red,
    mark=square,
    ]
    coordinates {
    (2,1.3334)(3,1.375)(4,1.409)(5,1.438)(6,1.462)(7,1.5)(8,1.5)(9,1.5)
    };
    \addlegendentry{Upper bound}
\addplot[
    color=blue,
    mark=o,
    ]
    coordinates {
    (2,1.3334)(3,1.3659)(4,1.35714)(5,1.35714)(6,1.35714)(7,1.35714)(8,1.35714)(9,1.3334)
    };
    \addlegendentry{Lower bound}
\end{axis}

\end{tikzpicture}
\caption{Known lower and upper bounds on the bin stretching factor}
\label{graph_bounds_1}
\end{figure}

As one can see in Figure~\ref{graph_bounds_1}, the gap between upper and lower bound for $m=3$ bins is pretty tight, but there is still much uncertainty about where $s^*$ really stands for $m\geq4$.
\\ \\
Let us now construct a rudimentary algorithm to try to find lower bounds as in \citet{Gabay2017}. To do so, we still need the following definitions:\\

\begin{definition}[Packing]
A \textit{packing} $b$ is defined by by a vector $b=(b_1, \dots, b_m)^T \in \mathbb{N}^m$ such that:
$$b_1\geq b_2 \geq \dots \geq b_m\geq 0$$
We also write: $$||b||_{1} = \sum_{i=1}^{m} b_i$$ 
\end{definition}
A packing corresponds to a filling of the bins. Since all bins are equivalent, we avoid symmetries by having the bins in a packing given by non increasing order.

Let us note $e_i$ the vector with a 1 on the $i$-th coordinate and 0 elsewhere.
If the bins have loads corresponding to a packing $b$, and an item $y$ is then placed into some bin $i$, we write the resulting packing $b' = b + y\cdot e_i$. This is an abuse of notation, as a packing is defined as having all coordinates non increasing. Doing this operation when considering packings as vectors does not necessarily preserve this property. However, we consider that this operation also rearranges the coordinates so that the result is still a packing.

A packing does not correspond to a unique state in the 2 player game; all the items that were previously given must also be considered. For example, with two bins of capacity 3, the packing $(2, 2)$ could be obtained after two objects of size 2, or after two objects of size 1 and one object of size 2. In the first case, the adversary can only send an item of size 1, while in the second, the adversary can send an item of size 2. So considering only packings isn't sufficient: each possible game state corresponds to both an instance (items that were sent) and a corresponding packing (loads of the bins). We call a configuration such a game state:
\begin{definition}[Configuration]
A \textit{configuration} is a pair $(I, b)$, with $I = (x_1, \dots, x_k)$ corresponding to the set of items that were given and $b$ corresponding to a packing of the bins. A configuration $(I, b)$ must verify:\begin{itemize}
    \item $(\frac{x_1}{g}, \dots, \frac{x_k}{g}) \in \mathcal{I}_{BS}$
    \item $b$ is a packing
    \item $b$ should be "reachable" with the set of items $I$, that is to say:
$$\exists \;(P_1, \dots, P_m) \text{ partition of }\{1, \dots, k\} \;|\; \forall i\in \{1, \dots, m\} \; \sum_{j\in P_i}x_j = b_i$$
\end{itemize}
\end{definition}
It can be remarked that this 2-player game has only two outcomes: either the algorithm or the adversary wins. This means that for any game state (i.e. for any configuration) there exists a player with a winning strategy. 
\\

We now have all the tools needed to construct a rudimentary foundation of an algorithm to obtain lower bounds. Algorithm~\ref{alg1} and Algorithm~\ref{alg2} form an example of what an implementation could look like. However, without adding efficient pruning, this algorithm will probably not be useful to obtain any new results. The two functions in the pseudo-code below simulate the players "adversary" (which sends items) and "algorithm" (which places items in bins).
\\
\begin{algorithm}[H]
	\caption{Function \textit{adversary\_to\_play}} 
	\label{alg1}
	\begin{algorithmic}[1]
	    \State \textbf{Input}: Configuration $C=(I, b)$
	    \State \textbf{Output}: Boolean value, True if the bound $t/g$ was proved, False otherwise
		\For {$y$ extension of $I$}
			\State $\nu \leftarrow$ \textit{algorithm\_to\_play}$(C, y)$
			\If{$\nu$ is True}
			    \State \textbf{Return} True
			\EndIf
		\EndFor
	    \State \textbf{Return} False 
	\end{algorithmic} 
\end{algorithm}

\begin{algorithm}[H]
	\caption{Function \textit{algorithm\_to\_play}}
	\label{alg2}
	\begin{algorithmic}[1]
	    \State \textbf{Input}: Configuration $C=(I, b)$ and item $y$
	    \State \textbf{Output}: Boolean value, True if the bound $t/g$ was proved when adding the item $y$ to the configuration $C$, False otherwise
		\For {bin $i = 1, 2, \ldots, m$}
		    \State Construct configuration $C' = (I\oplus y, b + e_i\cdot y)$
		    \If{$b_i + y < t$}
			    \State $\nu \leftarrow$ \textit{adversary\_to\_play}($C'$)
			    \If{$\nu = False$}
			        \State \textbf{Return} False 
			    \EndIf
			\EndIf
		\EndFor
		\State \textbf{Return} True
	\end{algorithmic} 
\end{algorithm}

In the previous algorithms, for simplicity, the method to find all possible extensions of an instance (line 3 of the first algorithm) isn't explained. This will however be detailed in Section~\ref{next_item}.

A simple improvement of the previous algorithms is to use a hash table to store all the configurations that were already explored. \citet{BOHM20221} proposed Zobrist hashing. On one hand, using hash tables yields a strong increase in search speed but on the other hand, this optimization may also require a large amount of memory, and especially so when the number of bins $m$ and the capacity $g$ increase.

The game trees on which such lower bounds are found have humongous size; lower bounds may also take long computation time to be found. We will propose in the next sections new methods to speed up the searches, in order to find better lower bounds. 

\section{Propagating algorithm winning packings}\label{section_propag_alg}
As explained previously, a game state (or configuration) is represented by both a packing (the loads of the bins) and an instance (the set of items that were given). It can be observed that if the player "algorithm" manages to reach certain packings, the adversary will never be able to win -  and hence we will not be able to prove the lower bound $\frac{t}{g}$. For example, with two bins ($m=2$) each of capacity $g = 6$ while trying to prove the bound $\frac{8}{6}$, if we reach the packing $b=(6, 5)^T$ then the adversary can only give an item of size $1$ which the algorithm can place in any of the bins without reaching $t=8$; the packing $(5, 5)^T$ is also always winning for the algorithm, since only one item of size $2$ or two items of size $1$ can be given, reaching the states $(7, 5)$ or $(6, 6)$.

Knowing if a packing is always winning for the algorithm allows us to stop exploring right away upon reaching this packing.

A packing $b$ is said to be \textit{algorithm-winning} if the algorithm always wins upon reaching this packing, regardless of what items were given beforehand by the adversary. We note $P_{alg}$ the set of algorithm-winning packings. More formally:
\begin{definition}[Algorithm-winning packing]

Let $b$ be a packing. If $\forall I \in \mathcal{I}_{BS}$ such that $(I, b)$ is a configuration, and that this configuration $(I, b)$ is winning for the algorithm, then $b\in P_{alg}$.
\end{definition}

We propose here a simple criterion to decide quickly for some packings if they are always winning for the algorithm. Stronger versions of this criterion were given by \citet{BOHM20221}, labelled as "good situations". However, we use this criterion for its simplicity and we will build on the idea to construct another property which is even stronger than what was proposed in \cite{BOHM20221}.
\begin{property}\label{init_palg}
Let $b$ be a packing. We remind that the coordinates of $b$ are sorted in non increasing order, so $b_1 \geq b_2 \geq \dots \geq b_{m-1} \geq b_m$.\\
$$\left(b_1<t\text{ and } b_{m-1} >  \frac{mg-t}{m-1}\right) \implies b\in P_{alg}$$
\end{property}

\begin{proof}
Let $b = (b_1, \dots, b_m)$ be a packing, with $b_1 < t$. Let us remark that if placing all the remaining items in $b_m$ does not reach at least $t$, then the bound cannot be proven; i.e. $b_m + mg-||b||_1 < t \implies b\in P_{alg}$. The lefthand side of the implication is equivalent to $mg-t < ||b||_1 - b_m $.\\Yet $||b||_1 - b_m = \sum_{i=1}^{m-1}b_i \geq (m-1) b_{m-1}$. So if $(m-1) b_{m-1} > mg-t$, then $b_m + mg-||b||_1 < t$; hence the bound cannot be proven, and $b\in P_{alg}$.
\end{proof}

We will provide an intuition of why this property is true through an example. Let us consider 2 bins of capacity 6 and let's try to prove the bound 8/6. It is possible to give a graphical representation of packings as in Figure~\ref{greensreds}.

\begin{figure}[htb]
    \centering
    \includegraphics[scale = 0.4]{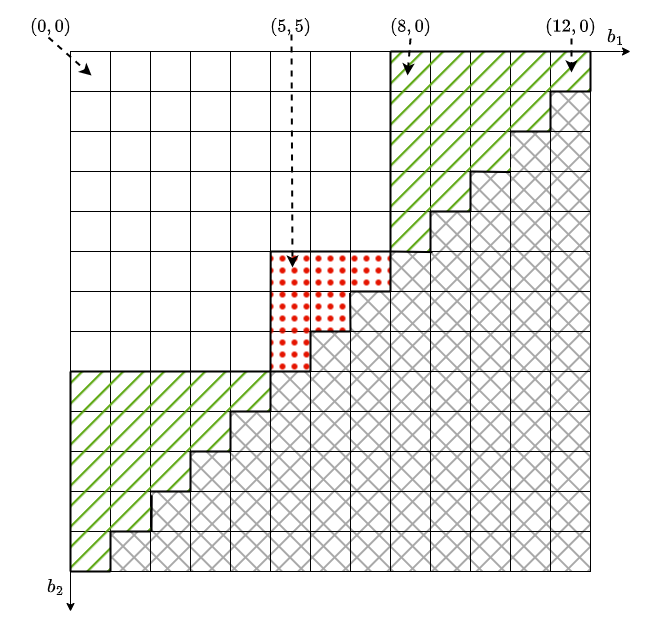}
    \caption{Geometric representation of packings, for $m=2$, $t=8$ and $g=6$ (red dotted area: algorithm winning packings, green hatched area: adversary winning packings, double hatched grey area: unreachable packings)}
    \label{greensreds}
\end{figure}

Each packing corresponds to a cell in the grid. The double hatched grey area corresponds to packings where the total load is strictly bigger than 12, which is impossible since all the items given must fit into $2$ bins of size 6. The green hatched areas correspond to packings where at least one bin has load over 8 - hence the adversary always wins when reaching a green packing. 
\\
From a certain packing, adding an item to any bin can only increases one coordinate of the packing. It can be seen in the figure that for any packing in the red dotted area, it will never be possible to reach a packing in the green hatched area. Hence the adversary is never able to prove the bound when a packing from this red dotted area is reached; so this red dotted area corresponds to algorithm-winning packings. 

\begin{figure}
    \centering
    \includegraphics[scale = 0.4]{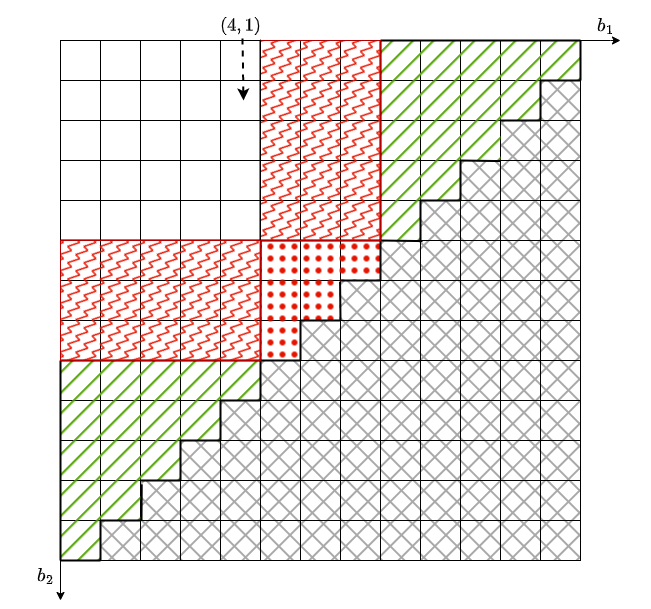}
    \caption{Extending the algorithm-winning packings to the area with red zig-zag lines}
    \label{greensreds2}
\end{figure}

Let us now take a look at the area with red zig-zag lines in Figure~\ref{greensreds2}. If a packing from this area is reached, then the algorithm can always decide for the next items to go in the direction of the red dotted area. Eventually, the red dotted area will be reached, and the algorithm will win. So this area with red zig-zag lines also corresponds to algorithm-winning packings.

This criterion isn't that strong, in the sense that it does not cover enough packings to significantly speed up the search for lower bounds, and especially so when the number of bins increases. We propose the following idea to increase the number of packings that are always winning for the algorithm: it is possible to propagate the algorithm-winning packings.
\\
\\
For a packing, if for all possible item sizes $s$ that can be given, it is always possible to put the item in a bin $i$ to reach an algorithm-winning packing, then the packing is also algorithm-winning. The biggest item size that could be given is smaller $\min\{g;\; mg - ||b||_1\}$, because we know that the sum of item sizes is smaller than $mg$ and that each item must individually be smaller than $g$.\\

Let us look at an example with $2$ bins of capacity $6$ while trying to prove the bound $8/6$. In Figure~\ref{greensreds2}, the packing (4, 1) can't be decided to be algorithm-winning with the previous criteria. However, if an item of size 1, 2 or 3 is given, this item could be placed in the first bin and an algorithm-winning packing, (5, 1), (6, 1) or (7, 1) would be reached. If an item of size 4, 5 or 6 is given, it could be placed in the second bin and once again, an algorithm-winning packing, (5, 4), (5, 5) or (6, 5) would be reached. Hence, no matter what the given item is, it is possible to place it in a certain bin so that an algorithm-winning packing is reached. This means that the packing (4, 1) is also algorithm-winning. Then, we can also use the new information that this packing (4, 1) is algorithm-winning to propagate even more.

Let us now properly state this idea through the following lemma:

\begin{property}\label{p_lemma}
Let $b$ be a packing, and $s_{max} = \min\{g;\; mg - ||b||_1\}$ be the maximum size of an item that could be given from this packing. Then:\\
$$( \,\forall s \in \{1, \dots, s_{max}\}, \; \exists\; 1\leq i\leq m \;\text{such that}\; b + s.e_i \in P_{alg}) \, \implies (b\in P_{alg})$$
with $e_i$ being the vector with a 1 on the $i$-th coordinate and 0 elsewhere.
\end{property}

\begin{proof}
Upon reaching a packing $b$ satisfying this condition, no matter what item is sent next, it is possible to place it in a bin to reach a new configuration where the algorithm always wins. Hence the algorithm also always wins from the packing $b$.
\end{proof}

This lemma allows to propagate algorithm-winning packings. We can first identify some algorithm-winning packings with the use of the Property~\ref{init_palg}, then apply Property~\ref{p_lemma} to propagate algorithm-winning packings.

\begin{figure}[htpb]
    \centering
    \includegraphics[scale=0.4]{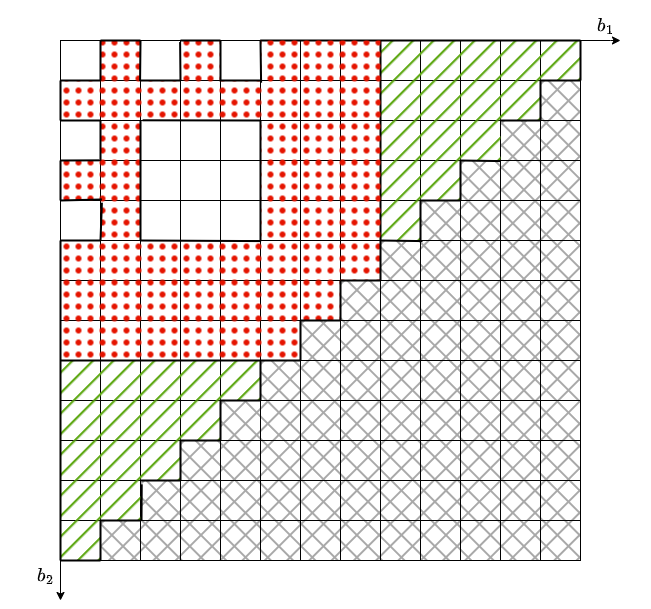}
    \caption{Propagation of algorithm-winning packings using Property~\ref{p_lemma}}
    \label{propag_example}
\end{figure}

Figure~\ref{propag_example} is an example of what the propagation looks like for $m=2$, $g=6$, $t=8$.

It is actually possible to improve this propagation idea: in some packings, the adversary may need to send a big object to "jump" over algorithm-winning packings; as such, during the tree-search, if this big object can't be sent by the adversary, then we already know that the algorithm wins. Moreover, it can be observed that, the further the game goes, the more the player adversary is restricted in the items he can give. Hence, if at some point the adversary can't send a certain size of item, then he will never be able to send this size for the rest of the game. \\
Using this idea, it is possible to propagate as well a required item size for each packing: if the player adversary, upon reaching a packing, can't send at least the required size of this packing, then the bound can't be proven - in other words, the algorithm always wins.

Here is an illustration of this idea. Let us look at Figure~\ref{propag_example}, and try to prove the bound $8/6$ with $2$ bins. If we consider the packing (4, 4), we need to send an item of size at least $4$ to not give the possibility to the algorithm to go into an algorithm-winning packing.\\
In the packing (4, 3) we need to send an item of size at least $5$.\\
In the packing (4, 2), at least $6$.\\
In the packing (3, 2), if we send an item of size between 2 and 5, the algorithm can decide to go to an algorithm-winning packing. So an item of size either 1 or 6 must be sent. However, if the adversary cannot send an item of size 6 and sends an item of size 1, the algorithm could choose to place the item in the first bin, thus reaching the packing $(4, 2)$, which requires an item of size 6. The adversary can't send an item of size 6 since he couldn't do that previously - so the algorithm wins.\\

We now properly present this idea, and also merge it with the previous propagation idea. We define the following function $v$ on packings such that if the adversary can't send at least $v(b)$ on the packing $b$, then the algorithm is winning. On packings which are known to be algorithm-winning packings (using the algorithm-winning propagation idea), we fix $v(b) = g+1$ to be consistent with the previous definition of $v$.
\begin{definition}[Minimum biggest extension]
For any packing $b$, we define $v(b)$ as:
\begin{itemize}
    \item If $b_1 \geq t$, $v(b) = 0$ \textit{(the bound $\frac{t}{g}$ is proved)}
    \item Else, if $||b||_1 = mg$ then $v(b) = g+1$ \\\textit{The adversary can't send any additional items, and the bound has not been proved, so the algorithm wins}
    \item Else, if $b_{m-1} \geq \lceil \frac{mg-t}{m-1} \rceil + 1$, $v(b) = g+1$\\\textit{(according to Property~\ref{init_palg}, this packing is algorithm-winning)}
    \item Otherwise, let $s_{max} = \min\{g;\; mg - ||b||_1\}$ be the maximum item size that could be given from the packing $b$, and:
    $$v(b) = \min_{s\leq s_{max}}\; \max_{i \leq m}\;\; \max \{s; v(b + s\cdot e_i)\}$$
\end{itemize}
\end{definition}

\begin{theorem}
If from the configuration $(I, b)$, the adversary is not able to send an item of size at least $v(b)$ then the algorithm wins.

In other words, if $(I, b)$ is a configuration and $v(b)\neq 0$, then:
$$v(b) \text{ is not an extension of I } \implies (I, b) \text{ is a winning configuration for the algorithm}$$
\end{theorem}

\begin{proof}
Let us prove this theorem by induction on $||b||_1$ non increasing.\\

If $||b||_1 = mg$, then the adversary cannot send any additional items so the 2-player game is over. Either $b_1 \geq t$, then $v(b) = 0$; or the largest bin has load strictly below $t$, so $v(b) = g+1$. In both cases, the theorem is true. 

Now, let $b$ be a packing with $v(b)\neq 0$ and let us suppose that this theorem was proved for all $b'$ such that $||b'||_1>||b||_1$. We then note $s_{max} = \min\{g;\; mg - ||b||_1\}$ to be the maximum item size that could be given from the packing $b$.

Let $I\in \mathcal{I}_{BS}$ be an instance such that $v(b)$ is not an extension of $I$. Let us prove that no matter what item is sent from the configuration $(I, b)$, there exists a bin such that placing the item in this bin results in a win for the player algorithm.

We define $y_{max} = \max\{y \text{ extension of } I\}$, the biggest item that can be sent from the configuration $C = (I, b)$. Since $v(b)$ is not an extension of $I$, we have $y_{max} < v(b)$.\\
Let us show that $y_{max} \leq \min\{g;\; mg - ||b||_1\} = s_{max}$.  Since any extension $y$ of $I$ must satisfy $I\oplus y \in \mathcal{I}_{BS}$, then: $\begin{cases}
||b||_1 + y \leq mg \Leftrightarrow y \leq mg - ||b||_1\\
y\leq g
\end{cases}$

So $y_{max} \leq \min\{g;\; mg - ||b||_1\} = s_{max}$.

Let $y\in \{1, \dots, y_{max}\}$. By definition, $v(b) = \min_{s\leq s_{max}}\; \max_{i \leq m}\;\; \max \{s; v(b + s\cdot e_i)\}$. Let us write $G_y(i) = \max \{y; v(b + y\cdot e_i)\}$, and $i_{max} = \argmax G_y$.\\
We then have: $\forall y\leq y_{max},\;\;\;v(b) \leq G_y(i_{max})$.\\

Then, we show that placing the item $y$ in the bin $i_{max}$ wins for the algorithm:\\
If $G_y(i_{max}) = y$, then $v(b) \leq G_y(i_{max}) = y \leq y_{max} < v(b)$, which is a contradiction. 
\\
So we have $G_y(i_{max}) = v(b + y\cdot e_i)$. Let us note the configuration $C' = (I\oplus y, b+y\cdot e_{i_{max}})$. Then: 

$v(b) \leq G_y(i_{max}) = v(b+y\cdot e_{i_{max}})$.

Since $v(b)$ is not an extension of $I$ and $v(b+y\cdot e_{i_{max}})\geq v(b)$ then $v(b+y\cdot e_{i_{max}})$ is not an extension of $I$ hence not an extension of $I\oplus y$.

By induction, as $v(b+y\cdot e_{i_{max}})$ is not an extension of $I\oplus y$, the algorithm always wins from the configuration $C'$.

So for any item $y$ that can be sent, placing the item in the bin $i_{max} =  \argmax G_y$ results in a configuration that is winning for the algorithm, which concludes the proof.
\end{proof}

To compute $v(b)$, we proceed by non increasing values of $||b||_{1}$, starting with $||b||_1 = mg$. 

To fully use the information in our tree search, we need to store for each possible packing $b$ the associated value $v(b)$. Hence we need to know how much memory to allocate in order to store all of the values $v(b)$, and thus we would like to know how many packings there are. If there wasn't the constraint forcing the coordinates to be non increasing, a possible upper bound of this number could be $t^m$ possible packings since each bin could have $t$ different possible values (from 0 to $t-1$). However, we can be much more efficient than this by using symmetries as follows.
\begin{notation}\
Let $k\in \mathbb{N}^*$ and $n\in \mathbb{N}^*$. We write $P_{k, n}$ the set of packings over $n$ bins such that each coordinate is strictly less than $k$ :
$$P_{k, n} = \{ (b_1, \dots, b_n)\in \mathbb{N}^n\;|\; k > b_1\geq b_2 \geq \dots \geq b_n \geq 0\}$$ We moreover define $N_{k, n}$ to be the cardinality of the previous set:
$$N_{k, n} = |P_{k, n}|$$

\end{notation}
We aim to compute $N_{t, m}$ to know how much memory is needed to store each value $v(b)$.
\begin{property}\label{Nkm}
Let $k\in \mathbb{N}^*$ and $n\in \mathbb{N}^*$. Then the number of packings over $n$ bins such that each coordinate is strictly less than $k$ is exactly:
$$N_{k, n} = \binom{k+n-1}{k-1} = \frac{k(k+1)(\dots)(k+n-1)}{n!} $$
\end{property}
The proof is given in Appendix~\ref{Proof_Nkm}.

With the previous result we can compute the number of packings $N_{t, m}$, and we know how much memory is needed to store all the values of $v(b)$. We now only need a bijection, which we will call $ind_m$, from $P_{t, m}$ to the integers between $0$ and $\binom{t+m-1}{t-1}-1$. With such a bijection, we can store all the values $v(b)$ in a table, such that the value $v(b)$ in the table is at the index $ind_m(b)$.
\begin{definition}[Index of a packing]
The index $ind_m(b)$ of a packing $b = (b_1, \dots, b_m)$ is defined as follows:
$$ind_m(b) = \sum_{i=1} ^ {m} N_{b_i, m-i+1}$$
\end{definition}
\begin{property}\label{bijection_prop}
For any $t$, the function $ind_m$ is a bijection between $P_{t, m}$ and $\left\llbracket 0, \binom{t+m-1}{t-1}-1\right\rrbracket$
\end{property}
The proof is given in Appendix~\ref{bijection_prop_proof}

This bijection allows an efficient memory usage to store the value $v(b)$ for each $b$. One can simply compute and store in a table all the coefficients $N_{t', m'}$ for $t'=1, \ldots, t$ and $m' = 1, \ldots, m$. Then, a table of size $N_{t, m}$ can be allocated to store the values of the function $v(b)$. To retrieve the information $v(b)$ of a packing, one just needs to compute $ind_m(b)$ which corresponds to the index in the table. This computation is done in $O(m)$.

It is possible to be even more precise than the previous analysis when counting the number of packings. The previous counting method didn't take into account the constraint on packings that $||b||_1 = \sum_{i=1}^m b_i \leq mg$. This will be very useful in Section~\ref{next_item}, to compute efficiently which items can be sent from a given configuration. 
\section{Knowing which item can be sent next}\label{next_item}
The player adversary can only send items such that all the items can fit into $m$ bins of size $g$ - hence we must compute which items can be sent from some state of the game. If an item can be sent, then according to Property~\ref{smaller_extensions}, any smaller item can also be sent. So the problem is equivalent to finding the largest item that can be sent from a game state: if $(I, b)$ is a configuration corresponding to a game state, we want to compute:
$$y_{max}(I) = \max\{y\text{ extension of } I\}$$

However, finding this value is a NP-hard problem (reduction from bin packing) and we have to compute this value at every game state. This makes it critical to optimize its computation.

Previously, this problem was solved by calling an optimization solver \citep{Gabay2017}. This has been improved in \citet{Bohm2017} where a combination of methods was used: lower and upper bounds, hash tables and dynamic programming. The authors claimed their approach was much faster than using a solver. We propose here an even better method than what was proposed so far, also using dynamic programming. 

\begin{notation}
Let $I$ be an instance. We write $\mathcal{B}_I = \{b \text{ packing } | (I, b) \text{ configuration and }b_1\leq g\}$.
\end{notation}
Computing the largest extension of an instance $I$ can be done by considering the set $\mathcal{B}_I$: 
\begin{property}\label{DPlemma}
The biggest extension $y_{max}(I)$ of an instance $I$ is:
$$y_{max}(I) = \max_{b\in \mathcal{B}_I} \{g-b_m\}$$
\end{property}
\begin{proof}
Let us note $I=(x_1, \dots, x_k)$.

For any packing $b\in \mathcal{B}_I$, the value $g-b_m$ is an extension of $I$, since the item $g-b_m$ can be placed in the bin $b_m$. So $$y_{max}(I) \geq \max_{b\in \mathcal{B}_I} \{g-b_m\}$$ 

Since $y_{max}(I)$ is an extension of $I$, then there is a partition of the objects $(x_1, \dots, x_k, y_{max}(I))$ in $m$ bins of capacity $g$. Let us construct the packing $b$ by packing only the items from $I$ according to the partition. Then, by definition $b\in \mathcal{B}_I$. Moreover, $y_{max}(I)$ can be placed in a bin with load $b_i \geq b_m$ of $b$. So $b_m + y_{max}(I) \leq b_i + y_{max} (I)\leq g$, and then:
$$y_{max}(I) \leq g-b_m$$
Since $b\in\mathcal{B}_I$, we have:
$$y_{max}(I) \leq \max_{b\in \mathcal{B}_I} \{g-b_m\}$$
which concludes the proof.
\end{proof}

To compute $y_{max}(I) $ we propose a dynamic programming method to construct the set $\mathcal{B}_I$. Let $I = (x_1, \dots, x_k)$ be an instance. Let us note $I_i = (x_1, \dots, x_i)$ for $i=1, \dots, k$ and $I_0 = \emptyset$ be the empty instance. To compute the largest extension of $I$, the dynamic program constructs the sets of packings $\mathcal{B}_{I_i}$ as follows:
\begin{itemize}
    \item $\mathcal{B}_{I_0} = \{(0, \dots, 0)\}$
    \item For $1\leq i \leq k$, $\mathcal{B}_{I_i} = \{b + x_i\cdot e_j \text{ such that }b\in \mathcal{B}_{I_{i-1}}, \;b_j + x_i \leq g \text{ and } j\in\llbracket1, m\rrbracket\}$
\end{itemize}
In other words, for any packing $b$ in $\mathcal{B}_{I_{i-1}}$, we look at all the packings formed when adding to $b$ the item $x_i$ into a bin. We only keep the packings where all bins have loads smaller than or equal to $g$. The set $\mathcal{B}_{I_{i}}$ can then be constructed by the union of all those packings. This allows the computation of the set $\mathcal{B}_I$. We can then find $\max_{b\in \mathcal{B}_I} \{g-b_m\}$ which is equal to the largest extension of $I$ by Property~\ref{DPlemma}.

This method, however, creates several problems when trying to actually implement it. 
\begin{enumerate}
    \item We do not know the cardinality of $\mathcal{B}_{I_i}$, which prevents us to allocate memory in an efficient manner.
    \item We need a way to prevent the same packing to be added to a set $\mathcal{B}_{I_i}$ several times. 
\end{enumerate}
To deal with the second problem, \citet{BOHM20221} use a hash table. However, we propose another idea. We first tackle the first problem by counting the exact number of packings in $\mathcal{B}_I$, which is the object of what follows; we then show that counting also results in a solution to the second problem. The approach is similar to the counting arguments given in Section~\ref{section_propag_alg}, but is a bit more complex and more accurate, since the constraint $||b||_1 \leq mg$ on packings is now taken into account. 

We first aim to compute the biggest number of packings that can be in the set $\mathcal{B}_I$ for any instance $I$.
\begin{notation}
Let $S\geq 0$, $k\geq 0$ and $n\geq 1$ be integers. We write $P_{S, k, n}$ the set of packings over $n$ bins such that each coordinate is smaller or equal to $k$ and the sum of the coordinates is exactly $S$:
$$P_{S, k, n} = \left\{(b_1, \dots, b_n)\in \mathbb{N}^n\; \text{such that } \begin{cases}k \geq b_1 \geq \dots \geq b_n\\||b||_1 = S \end{cases} \right\}$$
We also write $N_{S, k, n} = |P_{S, k, n}|$.
\end{notation}
For any instance $I$, $b\in \mathcal{B}_I \implies b\in P_{||I||_1, g, m}$. Hence $|\mathcal{B}_I| \leq |P_{||I||_1, g, m}| = N_{||I||_1, g, m}$.
\begin{property}\label{lemma_PSK_rec}
$$N_{S, k, 1} = \begin{cases} 1 \text{ if }k\geq S\\
0 \text{ otherwise}
\end{cases}$$
If $n>1$, then:
$$N_{S, k, n} = \sum_{i = 1}^{\min\{k, S\}} N_{S-i, i, n-1}$$
By convention, if $\min\{k, S\}\leq 0$ then the sum in the right hand side of the equation above is 0. 
\end{property}
\begin{proof}
With only 1 bin, if $k\geq S$, then $P_{S, k, 1} = \{(S)\}$. Otherwise, the set is empty.

If $n>1$:
We look at all the possibilies for the first coordinate of the packing, which corresponds to the variable $i$ in the equation. Then, the sum of the leftover coordinates must be equal to $S-i$, hence the formula. In a more formal way:

$$P_{S, k, n} = \bigcup_{i = 1}^{min\{k, S\}} \{(i, b_2, \dots, b_n) \text{ such that } (b_2, \dots, b_n)\in P_{S-i, i, n-1}\}$$
Since all sets in the union are disjoints, we obtain the expected result.
\end{proof}
We have not found an explicit formula for $N_{S, k, n}$ in the general case. However, Property~\ref{lemma_PSK_rec} still allows us to compute this coefficient rather efficiently. 
We propose to store all the coefficients $N_{S, k, n}$ for $S=0, \ldots, mg$, $k=0, \dots, g$ and $n = 1, \dots, m$ in a table. This only corresponds to $(mg + 1)(g+1)(m) = O(m^2g^2)$ coefficients which isn't large for the bounds we aim to improve: if $m= 10$, $g=200$, this corresponds to roughly $4\times10^6$ which is a very manageable number of integers to store for a computer.

Now that we know an upper bound of the cardinality of $\mathcal{B}_I$, we would like to address the second issue described earlier. We are constructing a set $\mathcal{B}_I$ by adding elements to this set one by one. To prevent duplicates, we want to construct a bijection between possible packings of this set and natural integers. 

\begin{property}\label{smart_bijection}
Let $S\in \mathbb{N}$. The function $f\;:\;P_{S, g, m}\rightarrow \llbracket0, N_{S, g, m}-1\rrbracket$ is constructed as follows:

If $S=0$, $f(b) = 0$.

Otherwise, let $b\in P_{S, g, m}$. Let $r$ be the index of the smallest bin of $b$ that is not empty: $b_r\neq 0$ and either $r = m$ or $b_{r+1} = 0$. Also note $S_i = S - \sum_{j = 1}^{i} b_j$.

Define $$f(b) = \sum_{ i = 0 }^{r-1} N_{S_{i}, b_{i+1} - 1, m-i}$$

Then $f$ is a bijection.

\end{property}

\begin{proof}
Let's define $<_{lex}$ the lexicographical order on packings. 
We will show that $f(b) = |\{b' \in P_{s, g, m} \;|\; b>_{lex}b'\}|$.

If $S = 0$ then only $b=(0, \dots, 0)$ is in the domain of $f$; $f(b) = 0$, and there are no packings smaller with the lexicographical order than $b$. 

Otherwise, for $1 \leq i \leq r-1$, the coefficient $N_{S_{i}, b_{i+1} - 1, m-i}$ corresponds precisely to the number of packings of $P_{S, g, m}$ that have exactly the first $i$ coordinates equal to the first $i$ coordinates of $b$ and that are stricly smaller lexicographically than $b$. Since the first $i$ coordinates are equal to those of $b$, the rest of the coordinates sums to $S_{i}$. Hence, $f$ is a bijection.
\end{proof}

Using this bijection, it is possible to check quickly whether or not a packing was already added to a set $\mathcal{B}_I$. We propose to have a boolean table of size $N_{||I||_1, g, m}$, initialized at False everywhere. Whenever we add a packing $b$ to $\mathcal{B}_I$, we first check whether or not this packing is already in the table by looking at index $f(b)$ of the boolean table. If the result is True, we already added $b$ in the set $\mathcal{B}_I$. Otherwise, we set the value in the boolean table to True, and we add the packing to the set.

During the search for a lower bound, at a certain configuration $(I, b)$, we keep in a table the set $\mathcal{B}_I$. When proposing an item $y$, we can quickly construct the new set $\mathcal{B}_{I\oplus y}$ from the previous $\mathcal{B}_I$. \citet{BOHM20221} had to construct the set $\mathcal{B}_{I\oplus y}$ from scratch everytime, while here we only have to do one iteration of the dynamic program.

\section{Results}\label{section-results}

Results have been computed on a computer with an Intel Core i5-4690 CPU and 15.6GiB of RAM. The code is available online\footnote{\url{https://github.com/lhommeant/Lower-bound-for-online-bin-stretching}}.

With the improvements presented in the previous sections, we managed to prove 3 new lower bounds:
\begin{theorem}
$$\text{For 6, 7 and 8 bins, }\frac{15}{11} \text{ is a lower bound of } s^*$$
\end{theorem}

For 6, 7 and 8 bins, the previous lower bound was $19/14\approx1.3571$ while the new bound is approximately $1.3636$.

We now compare the time needed by our method to prove some bounds.
The search method for lower bounds was improved and is faster in some cases than the implementation proposed by \citet{BOHM20221}. However, it is hard to compare times given by \citet{BOHM20221} with ours, because of the following reasons:

\begin{itemize}
    \item Most times given by \citet{BOHM20221} can't be compared, because they correspond to the search time while knowing some information on the solution (this information is called the monotonicity by \citet{BOHM20221}). This information makes it much easier to find a solution. 
    \item \citet{BOHM20221} used parallel computation.
    \item \citet{BOHM20221} used better hardware.
\end{itemize}

In Table~\ref{lower_bound_results} we compare our approach, for information purposes only, to the times given when this additional 'monotonicity' information isn't known. The row "Bound valid" determines if the corresponding lower bound is valid or not. We add a $(p)$ if the result from \citet{BOHM20221} was obtained with parallel computation. 

\begin{table}
\begin{center}
\begin{tabular}{ |c|c|c|c|c|c| } 
 \hline
 Bins & 3 & 3 & 3 & 3 & 4\\ 
 \hline \hline
 Lower bound & 19/14 & 30/22 & 55/40 & 56/41 & 19/14\\ 
 \hline
  Time from \citet{BOHM20221} & 2s. & 6s. & 3min. 6s. & 30min. & 18s.$(p)$\\
   \hline
   Our time & $<0.1$s & 0.1s & 5s. & 13s. & 0.1s\\ 
 \hline
 Bound valid & Yes & No & No & No & Yes\\
 \hline
\end{tabular}
\end{center}

\begin{center}
\begin{tabular}{ |c|c|c|c|c|c| } 
 \hline
 Bins & 4 & 4 & 6 & 7 & 8\\ 
 \hline \hline
 Lower bound & 30/22 & 34/25 & 15/11 & 15/11 & 15/11 \\ 
 \hline
  Time from \citet{BOHM20221} & 19s.$(p)$ & 48s$(p)$. & - & - & - \\
   \hline
   Our time & $18s$ & 2min. 24 & 14s. & 1min. 52s. & 1h. \\ 
 \hline
 Bound valid & No & No & Yes & Yes & Yes \\
 \hline
\end{tabular}
\end{center}
    \caption{Time comparison between \citet{BOHM20221} and our implementation}
    \label{lower_bound_results}
\end{table}

It can be seen in Table~\ref{lower_bound_results} that while our approach is much faster than the sequential algorithm by \citet{BOHM20221}, this is no longer true when comparing our implementation to their parallel one given by \citet{BOHM20221}. However, our algorithm is still able to be quite competitive. A possibility could be to parallelize our implementation in order to even further improve the bounds.

Lastly, Figure~\ref{improved_bds} shows our contribution to lower bounds in a graphical manner.

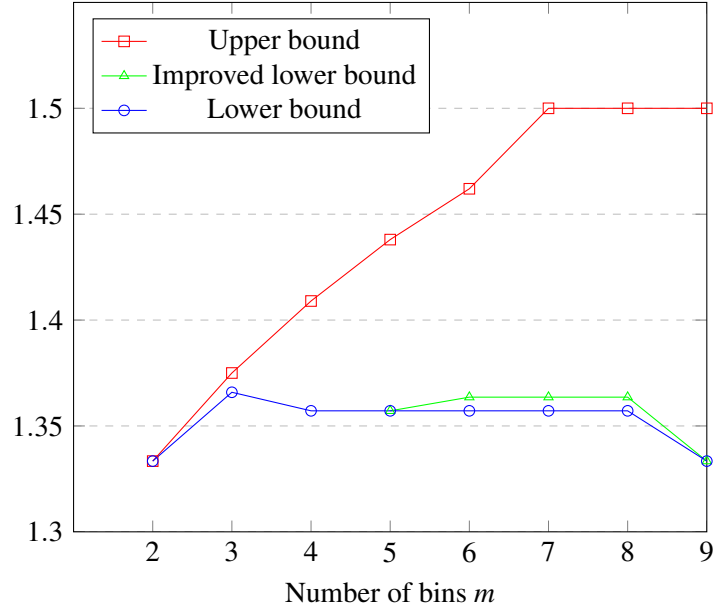
\begin{figure}[htbp]
\centering
\begin{tikzpicture}
\begin{axis}[
    xlabel={Number of bins $m$},
    xmin=1, xmax=9,
    ymin=1.3, ymax=1.55,
    xtick={2, 3, 4, 5, 6, 7, 8, 9},
    ytick={1.3, 1.35, 1.4, 1.45, 1.5},
    legend pos=north west,
    ymajorgrids=true,
    grid style=dashed,
]

\addplot[
    color=red,
    mark=square,
    ]
    coordinates {
    (2,1.3334)(3,1.375)(4,1.409)(5,1.438)(6,1.462)(7,1.5)(8,1.5)(9,1.5)
    };
    \addlegendentry{Upper bound}
\addplot[
    color=green,
    mark=triangle,
    ]
    coordinates {
    (5,1.35714)(6,1.3636)(7,1.3636)(8,1.3636)(9,1.3334)
    };
    \addlegendentry{Improved lower bound}
\addplot[
    color=blue,
    mark=o,
    ]
    coordinates {
    (2,1.3334)(3,1.3659)(4,1.35714)(5,1.35714)(6,1.35714)(7,1.35714)(8,1.35714)(9,1.3334)
    };
    \addlegendentry{Lower bound}

\end{axis}

\end{tikzpicture}
\caption{Improved lower bound}
\label{improved_bds}
\end{figure}

\section{Conclusion and perspectives}

This article presented an online problem and an approach based on computational search to find lower bounds on the performance of online algorithms for the problem.

New ideas were given to improve the method; we propagated information in order to speed up the search, and we improved the speed and efficiency memory-wise of the dynamic programming proposed in previous papers. This resulted in our search being faster than what was done previously, and in three new bounds being found: 15/11 for 6, 7 and 8 bins.

In order to improve the bounds of the online bin stretching problem even further, several things could be envisioned: our implementation could be parallelized, the algorithm to find bounds could be run for longer and on better computers and perhaps stronger pruning ideas similar to the ones presented here could be found.

This computational approach to find bounds could be generalized to many other online problems. One could start by considering problems closely related to online bin stretching: by changing the guarantee that the items fit into $m$ bins into another constraint, by changing the objective into the ratio of performance between an online algorithm and the optimal offline algorithm... Moreover, the proofs generated may give insight into creating general proofs: the lower bound $19/14$ seems to be valid in the general case but remains to be proven, if it is valid.

\appendix

\section{Proof of Property~\ref{Nkm}}\label{Proof_Nkm}
We want to prove that the number $N_{k, n}$ of packings over $n$ bins such that each coordinate is strictly less than $k$ is exactly $N_{k, n} = \binom{k+n-1}{k-1}$.

In order to prove Property~\ref{Nkm}, we need to use the following result, known as the hockey-stick identity:

\begin{lemma}[Hockey-stick identity]
Let $(n, r)\in\mathbb{N}^2$ such that $n\geq r$. Then:
$$\binom{n+1}{n-r} = \sum_{j=0}^{n-r}\binom{j+r}{r} = \sum_{j=0}^{n-r}\binom{j+r}{j}$$
\end{lemma}
The proof of this classical identity is not included here, but can be found online\footnote{\url{https://en.wikipedia.org/wiki/Hockey-stick_identity}}.

\begin{proof}
Let us now prove Property~\ref{Nkm}.

First, we remark that:
\begin{itemize}
    \item If $n = 1$ then $N_{k, 1} = k$
    \item Otherwise, $N_{k, n} = \sum_{i=0}^{k-1}N_{i+1,n-1}$
\end{itemize}
The second point is proved as follows:\\
$P_{k,n} = \bigcup_{b_1 = 0}^{k-1} \{(b_1, b_2, \dots, b_n) | b_1\geq b_2 \geq \dots \geq b_n \geq 0\}$\\
All sets in this union are disjoint. Hence:
$$|P_{k,n}| = \sum_{b_1=0}^{k-1} |\{(b_1, b_2, \dots, b_n) | b_1\geq b_2 \geq \dots \geq b_n \geq 0\}| = \sum_{i=0}^{k-1}N_{i+1,n-1}$$

With this remark, we can now prove Property~\ref{Nkm} by induction on $n$. If $n = 1$, then $N_{k, 1} = k = \binom{k+n-1}{k-1}$, so the property is true. 
\\
Let us now suppose that the property is true for $n-1$, and let's prove it for $n$. Then:

$$N_{k, n} = \sum_{i=0}^{k-1}N_{i+1,n-1} = \sum_{i=0}^{k-1} \binom{i+n-1}{i}$$
Applying the hockey stick identity, we obtain that:
$$ \sum_{i=0}^{k-1} \binom{i+n-1}{i} = \binom{k+n-1}{k-1}$$
So $N_{k, n} = \binom{k+n-1}{k-1}$.
\end{proof}
\section{Proof of Property~\ref{bijection_prop}}\label{bijection_prop_proof}
We prove here Property~\ref{bijection_prop}, which claims that the function $ind_m$ defined as $ind_m(b) = \sum_{i=1} ^ {m} N_{b_i, m-i+1}$ is a bijection between $P_{t, m}$ and $\left\llbracket 0, \binom{t+m-1}{t-1}-1\right\rrbracket$.
\begin{proof}
We show that the function $ind_m$ is a bijection by induction on $m$.
Let $t\in \mathbb{N}^*$.\\
If $m=1$, let $b\in P_{t, 1} = (b_1)$ then $ind_1(b_1) = N_{b_1, 1} = b_1$. So $ind_1$ is indeed a bijection between $P_{t, m}$ and $\left\llbracket 0, \binom{t+m-1}{t-1}-1\right\rrbracket$.

Let us suppose that $\forall m' < m$ the function $ind_{m'}$ is a bijection between $P_{t, m'}$ and $\left\llbracket 0, \binom{t+m'-1}{t-1}-1\right\rrbracket$.

Let $b\in P_{t, m}$. Then: $$ind_m(b) = \sum_{i=1} ^ {m} N_{b_i, m-i+1} = N_{b_1, m} + ind_{m-1}((b_2, \dots, b_m))$$

By convention, we write $\binom{n}{-1} = 0$.\\
Yet by induction, $ind_{m-1}$ is a bijection between $P_{b_1+1, m-1}$ and $\left\llbracket 0, \binom{b_1+m-1}{b_1}-1\right\rrbracket$. Let us note 
$$I_{b_1} = \left\llbracket \binom{b_1 + m - 1}{b_1 - 1}, \binom{b_1 + m - 1}{b_1-1}+\binom{b_1 + m - 1}{b_1} - 1\right\rrbracket$$
which simplifies to
$$I_{b_1} = \left\llbracket \binom{b_1 + m - 1}{b_1 - 1}, \binom{b_1 + m }{b_1} - 1\right\rrbracket$$
Then $\begin{cases} P_{b_1+1, m-1} \rightarrow I_{b_1}\\
(b_2, \dots, b_m) \mapsto ind_{m-1}(b_2, \dots, b_m)\end{cases}$ is also a bijection.

To conclude the proof, we just need to prove that all the sets $I_{b_1}$ are disjoints. Let us write $h(b_1) = \binom{b_1 + m - 1}{b_1 - 1}$.

Then $I_{b_1} = \llbracket h(b_1), h(b_1+1) - 1\rrbracket$. To prove that the sets are disjoints we then just need to prove that the function $h$ is strictly increasing:
\begin{itemize}
    \item If $b_1>0$, then $\frac{h(b_1 + 1)}{h(b_1)} = \frac{b_1 + m}{m} > 1$ 
    \item If $b_1 = 0$, then $h(0) = 0 < h(1) = 1$
\end{itemize}
So $h$ is strictly increasing. So the sets $I_{b_1}$ are all distinct.

So the function $\begin{cases} P_{t, m} \rightarrow \left\llbracket 0, \binom{t+m-1}{t-1}-1\right\rrbracket\\
b \mapsto ind_{m}(b)\end{cases}$ is a bijection.
\end{proof}

 \bibliographystyle{elsarticle-num-names.bst} 
\bibliography{refs.bib}

\end{document}